\documentclass[11pt]{article}

\usepackage[T1]{fontenc}
\usepackage[utf8]{inputenc}
\usepackage{lmodern}
\usepackage{newunicodechar}
\newunicodechar{’}{\textquoteright}
\usepackage{xcolor}
\usepackage{tikz}
\usepackage{amsmath,amssymb,amsthm,mathtools,booktabs}
\usepackage[a4paper,margin=1in]{geometry}
\usepackage{enumitem}
\newtheorem*{theorem*}{Theorem}
\usepackage[numbers]{natbib}
\usepackage[colorlinks=true,linkcolor=blue,citecolor=blue,urlcolor=blue]{hyperref}
\usepackage{todonotes}
\usepackage{comment}
\newtheorem{theorem}{Theorem}

\newtheorem{lemma}[theorem]{Lemma}
\newtheorem{corollary}[theorem]{Corollary}

\usepackage{algorithm}
\usepackage{algpseudocode}
\usepackage{thmtools}
\usepackage{thm-restate}
\newcommand{\ceil}[1]{\lceil {#1} \rceil}

\title{Improved Space Bounds for Promise Graph Coloring in Adversarial Streams}
\author{Sergio Thoumi\\Lebanese American University}
\date{}
\usepackage{microtype}

\newcommand{\Prb}{\mathbb P}

\begin{document}

\maketitle

\begin{center}
\textit{In loving memory of Carl Feghali}
\end{center}

\begin{abstract}
 We study promise graph coloring in one pass insertion-only adversarial streams. We give a randomized algorithm that distinguishes 3-colorable graphs from graphs that are not 6-colorable using \(O(n^{3/2}\log(n/\varepsilon)\log n)\) bits of space. This provides a truly subquadratic space algorithm for a constant coloring gap, answering an open problem posed by Assadi, Sundaresan, and Yazdanyar [SODA 2026] for $q=3$.

Combining this algorithm with the vertex-sampling lemma of Assadi et al.\ yields, for any $144\lceil\log_2 n\rceil\le Q\le n$, a one-pass decision algorithm distinguishing $3$-colorable graphs from graphs that are not $Q$-colorable using $O((n/Q)^{3/2}\log^5 n)$ bits of space. This improves upon the $\widetilde{O}((n/Q)^2)$ bound achievable via dynamic streaming techniques. Setting $Q=\lceil n^{1/3}\rceil$ yields a semi-streaming decision algorithm, answering the decision version of another open problem posed by Assadi et al.

On the negative side, we prove that distinguishing \(\chi(G)\le q\) from \(\chi(G)\ge(q-1)^2\) requires \(\Omega(n^2/q^5)\) bits of space, improving upon the previously established \(q^2/4\) threshold using the same space complexity.
\end{abstract}

\section{Introduction}

Graph coloring is one of the fundamental problems in graph theory. A well-studied formulation is the $k$-coloring problem where the input is a graph and one should decide\footnote{In the non-decision variant, one would find a coloring, if it exists.} whether the vertices of the graph can be colored with at most $k$ colors such that no two adjacent vertices share the same color. In the case where such a coloring is possible, the graph is said to be $k$-colorable. The number of colors needed is its \emph{chromatic number} and is denoted as $\chi(G)$.

One generalization of $k$-coloring has been introduced \cite{PromiseColoring} as a Promise Constraint Satisfaction problem, where, given a graph, one has to decide whether it is $k$-colorable or its chromatic number is at least $K$ under the promise that one of these cases holds (here, $k<K$). When $K=k+1$, this is simply the decision variant of the $k$-coloring problem. 

More recently, the promise variant of graph coloring has been studied \cite{assadi2026coloring} in the graph streaming model \cite{feigenbaum2005graph}. In that model, we have a set of $n$ vertices labeled $\{1,2,\ldots,n\}$ and a stream of edge additions (and deletions). The goal is to make one or a few passes over the edge stream and find a solution while using the least amount of memory possible. Obviously, any graph problem can be trivially solved using $O(n^2)$ space by simply storing all its edges.

Assadi et al. \cite{assadi2026coloring} studied the problem in multiple streaming settings. In this paper, all our results hold in the one pass, insertion-only adversarial-order streaming model \cite{guha2009stream}, in which an adversary knows the algorithm and the entire stream of edge insertions and may arrange them in any order (note that our lower bound in Section \ref{sec:lowerbound} consequently also applies to dynamic streams).

In that setting, Assadi et al. proved that distinguishing between $q$-colorable graphs and graphs with chromatic number at least $q^2/4$ requires $\Omega(n^2/q^5)$ space. In the two-party one-way communication model, the corresponding problem admits an \(O(1)\)-bit protocol when the NO threshold exceeds \(q^2\). We prove that this lower bound holds for distinguishing between $q$-colorable graphs and graphs with chromatic number at least $(q-1)^2$. This result is presented in Section \ref{sec:lowerbound}.

In Section~\ref{sec:abovesix}, we give a randomized algorithm
that distinguishes $\chi(G)\le3$ from $\chi(G)>6$ using
\[
  O\!\left(n^{3/2}\log(n/\varepsilon)\log n\right)
\]
bits of space. The algorithm always provides a coloring for 3-colorable
graphs and rejects graphs with chromatic number greater
than $6$ with probability at least $1-\varepsilon$.
For constant error, this gives a truly subquadratic-space
algorithm for a constant coloring gap, answering the $q=3$
case of Open Problem~1 of Assadi et al.

In Section~\ref{sec:generalupperbound}, we combine the decision variant of this algorithm with the vertex-sampling lemma of Assadi et al.
For every fixed integer $c\ge1$ and every integer
$144\lceil\log_2 n\rceil\le Q\le n$, we distinguish
$\chi(G)\le3$ from $\chi(G)>Q$ using
\[
  O\!\left((n/Q)^{3/2}\log^5 n\right)
\]
bits of space, with perfect completeness and success
probability at least $1-n^{-c}$.
This improves the $\widetilde{O}((n/Q)^2)$ upper bound
obtained from their dynamic-stream algorithm.
For all sufficiently large $n$, setting
$Q=\lceil n^{1/3}\rceil$ gives an $O(n\log^5 n)$-space
decision algorithm. Thus, we answer affirmatively the
decision part of their Open Problem~2.

\section{Preliminaries}
\label{sec:preliminaries}

For an integer $r\ge1$, let $[r]=\{1,\ldots,r\}$, and let
$\binom{W}{2}$ denote the set of unordered pairs of distinct elements
of a set $W$. We write $\log$ for the base-two logarithm. The notation $\widetilde{O}(f)$ suppresses factors polynomial in $\log n$.

\subsection{Graphs and colorings}

All graphs are finite, undirected, unweighted, and simple. Thus, a graph
$G=(V,E)$ has $E\subseteq\binom{V}{2}$, with no self-loops or parallel
edges. We write $n=|V|$, $m=|E|$, and abbreviate the unordered pair
$\{u,v\}$ by $uv$. For $H\subseteq V$, the induced subgraph on $H$, denoted $G[H]$, has $H$ as its set of vertices, and any edges of $G$ whose both endpoints are in $H$.

For $v\in V$, its open neighborhood is
$N_G(v)=\{u\in V:uv\in E\}$.

An edge partition into subgraphs $H_1,\ldots,H_t$ means that their
edge sets are pairwise disjoint and have union $E(G)$; their vertex
sets may overlap.

A proper $q$-coloring is a map $\varphi:V\to[q]$ such that
$\varphi(u)\ne\varphi(v)$ for every $uv\in E$.
The chromatic number $\chi(G)$ is the minimum number of colors in
a proper coloring.
A clique is a set of pairwise adjacent vertices, whereas an independent
set contains no edges. We use standard graph-theoretic terminology
as in~\cite{diestel2025graph}.

\subsection{Paths, cycles, and parity}

A path $P=(v_0,\ldots,v_r)$ consists of distinct vertices with $v_{j-1}v_j\in E$ for every $j\in[r]$. Its length is the number $r$ of edges.  A cycle of length $r\ge3$ consists of a path $(v_0,\ldots,v_{r-1})$ of length $r-1$ along with the edge $v_{r-1}v_0$.

A path or cycle is \emph{odd} or \emph{even} according to its length.
Paths and cycles need not be induced.

A graph is bipartite if its vertex set can be partitioned into two
independent sets. Equivalently, it admits a proper coloring with at
most two colors, or contains no odd cycle \cite{diestel2025graph}.

A forest is an acyclic graph. Each pair of vertices in the same
component of a forest has a unique connecting path. A forest on $n$ vertices has at most $n-1$ edges.

\subsection{Promise graph coloring}

For integers $1\le q<Q$, the promise decision problem asks an
algorithm to distinguish
\[
 \chi(G)\le q
 \qquad\text{from}\qquad
 \chi(G)\ge Q.
\]
The first case is a YES-instance and the second is a NO-instance.
The input is promised to satisfy one of these conditions; either
answer is allowed when $q<\chi(G)<Q$. Throughout this paper, we sometimes refer to the problem as $\chi(G)\le q \text{ vs } \chi(G)\ge Q$ coloring.

\subsection{Communication complexity and streaming lower bounds}

We use randomized two-party one-way communication protocols with
public randomness. Alice and Bob receive separate inputs and have
access to shared random bits independent of those inputs. Alice
sends a single message to Bob, who then outputs the answer.
Local computation is unrestricted. The communication cost is the
maximum message length in bits over all inputs and random choices;
the public random bits themselves do not count toward this cost. See~\cite{roughgarden2016communication} for background.

\paragraph{The INDEX problem.}
In $\mathrm{INDEX}_t$, Alice receives $x\in\{0,1\}^t$ and Bob
receives an index $i\in[t]$. Bob must output $x_i$.
Every randomized one-way protocol that succeeds with probability
at least $2/3$ on every input uses $\Omega(t)$ bits, even with public
randomness~\cite[Section~2.4]{roughgarden2016communication}.

\paragraph{The graph coloring communication problem.}
For integers $n\ge1$ and $1\le q<Q$, in
$\mathrm{COLOR}_{n,q,Q}$, Alice and Bob receive disjoint edge sets
$E_A,E_B\subseteq\binom{[n]}{2}$. The parameters $n,q,Q$ are
known to both players. For $G=([n],E_A\cup E_B)$, Bob must output
a YES-instance when $\chi(G)\le q$ and a NO-instance when
$\chi(G)\ge Q$, under the promise that one of these cases holds.

If a randomized algorithm accepts every YES-instance with probability $1$, it is said to have \emph{perfect completeness}.

\paragraph{Streaming and communication complexity.}
An $s$-bit one-pass streaming algorithm for this promise problem
induces a one-way protocol communicating at most $s$ bits with the
same error probability (See \cite[Proposition~4.2, arXiv version]{{assadi2026coloring}} for a proof).

Consequently, a one-way communication lower bound for
$\mathrm{COLOR}_{n,q,Q}$ yields the same streaming space lower bound.

\section{$\chi(G)\leq q$ vs $\chi(G)\geq (q-1)^2$ Coloring }\label{sec:lowerbound}

We first introduce the following Lemma \cite[Lemma~4.4 (arXiv version)]{assadi2026coloring}:

\begin{lemma}\label{lemma:subgraphscontruction}
For infinitely many integers $n \geqslant 1$ and any integer $1 \leqslant k \leqslant \sqrt{n}$, there exists a $k$-colorable graph $G=(V,E)$, such that its edges can be partitioned into
\[
t=\Omega\!\left(\frac{n^2}{k^5}\right)
\]
subgraphs,
\[
H_1 \cup \dots \cup H_t.
\]
The subgraphs have the following property: for every $i \in [t]$, the induced subgraph of $G$ on $V(H_i)$ is a vertex-disjoint union of $k$ separate $k$-cliques.
\end{lemma}

We recall the nomenclature that are also used in the cited construction that are used below.
The base graph has $k$ layers $L^{(1)},\ldots,L^{(k)}$.
Layer $L^{(r)}$ consists of groups
$L^{(r)}_j=\{v^{(r,j)}_1,\ldots,v^{(r,j)}_k\}$,
where $j\in[n/k^2]$. For
$b\in[n/(2k^2)]$ and $p\in[n/(2k^3)]$, the $z$-th line is
\[
 P_{b,p,z}=
 \bigl(v^{(1,b)}_z,v^{(2,b+p)}_z,\ldots,
       v^{(k,b+(k-1)p)}_z\bigr),
 \qquad z\in[k].
\]
The subgraph $H_{b,p}$ consists of the $k$ cliques on these lines.
Each clique is ordered by layer.
The index $z$ is a global vertex label, so every base-graph edge
joins vertices with the same $z$ and different layer indices.
Consequently, $\varphi(v^{(r,j)}_z)=r$ is a proper $k$-coloring.
This fixes both the common clique numbering and the common
vertex ordering used in the reduction.

Given a protocol for $\mathrm{COLOR}_{n,k+1,k^2}$, we create a protocol for $\mathrm{INDEX}_t$ as follows.

Let $H_1 \cup \dots \cup H_t$ be the subgraphs of the graph obtained from Lemma~\ref{lemma:subgraphscontruction} with $t=\Omega(n^2/k^5)$. Since the graph is $k$-colorable, fix a proper $k$-coloring $\varphi \colon V \to [k]$ such that each layer is assigned exactly one distinct color. The protocol then goes as follows:

\begin{algorithm}[H]
\caption{The protocol $\pi_{\mathrm{INDEX}}$ given the protocol
$\pi_{\mathrm{COLOR}}$}
\label{alg:index-to-color}

\begin{enumerate}[label=\arabic*.]
\item 

Given $x\in\{0,1\}^t$, for every $j\in[t]$, Alice adds the edges of $H_j$ to $E_A$ if and only if $x_j=1$. She then defines $G_A=(V,E_A)$.

\item  Let $D_i$ be the set of all possible edges between vertices  of $H_i$ that have been assigned a different color under $\varphi$. That is, $D_i \coloneq \{uv \colon u,v\in V(H_i),\ \varphi(u)\neq \varphi(v)\}$.

Additionally, let $L_i$ be the set of all possible edges between vertices of $H_i$ that have been assigned the same color under $\varphi$. That is, $L_i \coloneq \{uv \colon u,v\in V(H_i),\ \varphi(u)=\varphi(v)\}$.

Given an index $i\in[t]$, Bob defines his graph as $G_B=(V,E_B)$, where $E_B=L_i\cup (D_i\setminus E(H_i))$.

\item Alice runs $\pi_{\mathrm{COLOR}}$ on $G_A$ and sends the message $M(G_A)$ to Bob, and Bob outputs $x_i=1$ if $\pi_{\mathrm{COLOR}}$ declares $\chi(G)\geq k^2$ and otherwise outputs $x_i=0$.

\end{enumerate}
\end{algorithm}

\begin{theorem}\label{theorem:main}
Let $G$ be the graph produced by the reduction from an
instance $(x,i)$ of $\mathrm{INDEX}_t$ to
$\mathrm{COLOR}_{n,k+1,k^2}$. If $x_i=0$, then
$\chi(G)\leqslant k+1$. If $x_i=1$, then $\chi(G)\geqslant k^2$.
\end{theorem}

\begin{proof}

First, suppose $x_i=1$. Then Alice contributes all edges of $H_i$.
Together with Bob's edges $L_i\cup (D_i\setminus E(H_i))$, this makes
$G[V(H_i)]$ a clique of size $k^2$. Hence $\chi(G)\ge k^2$. 

We now show that the graph is $(k+1)$-colorable when $x_i=0$.

For each $z\in [k]$, define
\[
C_z \coloneqq \{1,2,\dots,k+1\}\setminus \{z\}.
\]

To color $G$, we first take a subgraph $H_j \neq H_i$ and color its $z$-th clique with
the colors in $C_z$. Since every $k$-clique is colored with $k$ distinct colors, this gives a proper coloring of the subgraph.

We repeat this process for every subgraph other than $H_i$, preserving both a common ordering of the cliques and a common internal ordering of vertices within each clique. That is, for every $r \in [k]$, the $r$-th vertex of the $z$-th clique in $H_j$ receives the same color as the $r$-th vertex of the $z$-th clique in $H_{j'}$.

Since the edges of $G$ can be partitioned into $t$ subgraphs, every edge of $G$ belongs to some subgraph\footnote{This excludes the edges $E_B$ which are outside of the construction of Lemma \ref{lemma:subgraphscontruction}.}. Therefore, if a single fixed vertex-coloring properly colors each subgraph
$H_j \neq H_i$, then every edge of $G$ that does not have an endpoint in $H_i$ has endpoints with distinct colors.

It remains to show that we can color the vertices of $H_i$ while only using colors from $\{1,2,\dots,k+1\}$ and maintaining the proper coloring of the graph.

Observe that by definition, each set $C_z$ misses exactly one color from $\{1,2,\dots, k+1\}$, and that is the color $z$.  For each $z \in [k]$, define $M_z \coloneqq z$.

To color $H_i$, we take the vertices of its $z$-th line (i.e. what was the $z$-th clique in the original graph) and color them with the color $M_z$. Note that vertices that also belong to other subgraphs would already be colored. In that case, we simply recolor them while ignoring the previously assigned color.

When $x_i$ = 0, the edges of $E(H_i)$ are absent from Alice's graph. Since $L_i$ and $D_i$ partition all vertex pairs in $V(H_i)$ and every clique edge of $E(H_i)$ joins vertices from different layers (so $E(H_i)\subseteq D_i$), Bob's edge set $E_B = L_i \cup  (D_i \setminus E(H_i))$ contains all pairs inside $V(H_i)$ except the pairs that lie on the same original line (clique) of $H_i$.

Therefore, every edge in $E_B$ inside $V(H_i)$ joins vertices from two different lines. Since the coloring gives line $z$ the color $z$, the coloring satisfies the constraints imposed by $E_B$.

It remains to show that the coloring satisfies the constraints imposed by $E_A$. Let $S_i\subseteq E_A$ be the set of all edges in $G$ with one endpoint in $V(H_i)$ and the other endpoint in $V(G_A) \setminus V(H_i)$. That is, $S_i \coloneq \{uv \in E_A \colon u \in V(H_i) \text{ and } v \in V(G_A) \setminus V(H_i)\}$. Due to the inducedness property of the original construction, $E_A\setminus S_i$ has no edges with an endpoint in $H_i$, so $S_i$ is exactly the subset of $E_A$ that imposes constraints on the coloring of $H_i$.

Now observe that $S_i$ has only edges whose endpoints are in the same line/clique number. This means that if we have an edge $xy\in S_i$ then both its vertices $x$ and $y$ must be in the $z$-th line/clique for some $z\in [k]$.

Since $x \in V(H_i)$ and $y \in V(H_j)$ for some $j \neq i$, then we know $x$ has color $M_z$ and $y$ has a color that is an element of $C_z$. By definition, $M_z\notin C_z$. Thus, the coloring also satisfies the constraints of $S_i$ and $G$ can be properly colored using $k+1$ colors.

\end{proof}

By letting $q=k+1$, we obtain the following result:

\begin{corollary}
For infinitely many integers $n \geq 1$ and any integer $3 \leq  q\leq \sqrt{n} + 1$, any randomized one-pass insertion-only streaming algorithm that distinguishes graphs with chromatic number at most $q$ from graphs with chromatic number at least $(q-1)^2$, with success probability at least $2/3$, requires $\Omega(n^2/q^5)$ bits of space in adversarial-order streams.
\end{corollary}

\section{$ \chi(G)\le 3\text{ vs }\chi(G)>6$ Coloring} \label{sec:abovesix}

\begin{theorem}\label{thm:core}
For every $n\ge2$ and $0<\varepsilon<1$, there is a randomized one pass adversarial-order insertion-only algorithm that distinguishes
\[
 \chi(G)\le3 \qquad\text{from}\qquad \chi(G)>6
\]
using $O(n^{3/2}\log(n/\varepsilon)\log n)$ bits in the worst case. It succeeds with probability at least $1-\varepsilon$.
\end{theorem}

In the above, a YES-instance implies that $\chi(G)\le3$ while a NO-instance implies that $\chi(G)>6$.

\vspace{5pt}

The main idea to prove Theorem \ref{thm:core} is to sample a subgraph $H$ such that every edge in $G$ is either in $H$ or has an odd-length replacement path in it. This path allows us to prove that any 3-coloring of $H$ results in a $6$-coloring of $G$, by showing that edges in $E(G\setminus H)$ can be properly colored (see Lemma \ref{lem:6color}).

\vspace{5pt}

\noindent\textbf{Construction.} We first describe a construction for an arbitrary fixed anchor set $S\subseteq[n]$, $|S|=\ceil{\sqrt{n}}=s$. Each vertex $v$ has an initially undefined anchor assignment $a(v)$. Once defined, the assignment is permanent. Let
\[
 A_a=\{v:a(v)=a\}\quad(a\in S),\qquad
 U=\{v:a(v)\text{ remains undefined}\}.
\]
The groups are obviously disjoint. 

Maintain a buffer $B$ and forests $F_{ab}$ indexed by unordered pairs of anchors, allowing $a=b$. For each arriving edge $\{u,v\}$, perform the below steps in the stated order:
\begin{enumerate}
    \item Check the assignment of each endpoint \emph{before processing this edge}. If either endpoint is unassigned, put the edge in $B$. Otherwise, send it to $F_{a(u),a(v)}$.

    \item For an edge $\{u,v\}$ sent to a forest $F_{a(u),a(v)}$, proceed as follows:
    \begin{enumerate}
        \item If $u$ and $v$ belong to different connected components, retain the edge in the forest.
        
        \item Otherwise, let $P$ be the unique path connecting $u$ and $v$ in the forest.
        \begin{enumerate}
            \item If $a(u)=a(v)$ and $P$ has even length, then we know there is an odd cycle in the neighborhood of the anchor. Halt the algorithm and return a NO-instance.
            \item Otherwise, we already have a odd-length replacement path, so we discard the edge.
        \end{enumerate}
    \end{enumerate}

    \item Assign each still-unassigned endpoint to the other endpoint if that other endpoint belongs to $S$.
\end{enumerate}

If no early NO-instance occurs, define the retained subgraph
\[
 H=\left(V, B\cup\bigcup_{a,b\in S}F_{ab}\right)\subseteq G(V,E).
\]
Find a proper 3-coloring $D:V\to\{1,2,3\}$ of $H$ by exhaustive search, and return a NO-instance if none exists. 

\begin{lemma}[Odd replacement paths]\label{lem:paths}
If the construction does not halt and return a NO-instance, every discarded edge of $G$ has assigned endpoints $u\in A_a$, $v\in A_b$, and an odd-length path from $u$ to $v$ in $F_{ab}$. All vertices of that path belong to $A_a\cup A_b$.
\end{lemma}

\begin{proof}
Let $\{u,v\}$ be a discarded edge. Since it was not buffered, both endpoints were already assigned when the edge arrived. Write $a(u)=a$ and $a(v)=b$. Hence the edge was sent to $F_{ab}$.

By step 2 of the algorithm, an edge sent to $F_{ab}$ is discarded only when its endpoints are already connected in $F_{ab}$. Thus, when $\{u,v\}$ is discarded, there is a unique path from $u$ to $v$ in the forest $F_{ab}$.

If $a\ne b$, every edge of $F_{ab}$ has one endpoint in $A_a$ and the other in $A_b$. Hence $F_{ab}$ is bipartite with sides $A_a$ and $A_b$. Since $u\in A_a$ and $v\in A_b$, their path in $F_{ab}$ has odd-length.

If $a=b$, step 2 would halt and return a NO-instance if the unique path between $u$ and $v$ in $F_{aa}$ had even length. Since the algorithm does not halt, this path must therefore have odd-length.

Finally, by construction, every vertex appearing in $F_{ab}$ belongs to $A_a\cup A_b$. Hence every vertex on the replacement path belongs to $A_a\cup A_b$.
\end{proof}

\begin{lemma}\label{lem:6color}
If the algorithm does not halt and return a NO-instance, every proper 3-coloring $D$ of $H$ yields a proper 6-coloring of $G$.
\end{lemma}
\begin{proof}
Fix any function $f:\{1,2,3\}\to\{1,2,3\}$ with $f(j)\ne j$, for example $f(j)=(j\bmod 3)+1$. Define
\[
 C(v)=
 \begin{cases}
  (D(a(v)),D(v)),&v\notin U,\\
  (f(D(v)),D(v)),&v\in U.
 \end{cases}
\]
Since every assignment edge $\{v,a(v)\}$ belongs to $H$ and $D$ properly colors $H$, we have
\[
D(a(v))\ne D(v)
\]
for every assigned vertex $v$. For $v\in U$, we also have $f(D(v))\ne D(v)$ by the definition of $f$. Hence every value of $C$ belongs to the six-element palette $\{(i,j):i,j\in\{1,2,3\},\ i\ne j\}$.

Observe that every edge of $H$ has endpoints with different second coordinates under $C$, since the second coordinate of $C(v)$ is $D(v)$ and $D$ properly colors $H$. In particular, this covers every edge with at least one endpoint in $U$, since all such edges are buffered and therefore belong to $H$.

Now consider a discarded edge $\{u,v\}$, with anchors $a,b$, respectively. If $D(a)\ne D(b)$, its first coordinates differ. Otherwise write $D(a)=D(b)=i$ and use the odd path from Lemma~\ref{lem:paths}. In particular, every path vertex $w$ is assigned to $a$ or $b$ and has its assignment edge in $H$. Hence, $D(w)\ne i$. The path is properly colored by $D$ using only the other two colors. Since the assignment edges $ua$ and $vb$ are in $H$, the path's odd-length forces $D(u)\ne D(v)$, so the second coordinates differ. This proves that $C$ is proper on every edge of $G$.
\end{proof}

\begin{lemma}\label{lem:complete}
If $G$ is 3-colorable, the algorithm never returns a
NO-instance and $H$ has a proper 3-coloring.
\end{lemma}

\begin{proof}
Suppose the algorithm terminates early with a NO-instance. Then, the arriving edge completes an odd cycle within some $A_a\subseteq N_G(a)$. In a proper 3-coloring, every neighbor of $a$ must avoid the color of $a$, so $G[A_a]$ must be bipartite. Such an odd cycle is therefore impossible.

Finally, since $H\subseteq G$, any proper 3-coloring of $G$ is also a proper 3-coloring of $H$. Hence, the exhaustive search cannot return a NO-instance either.
\end{proof}

\subsection{Space Complexity}
Sample $S$ uniformly among all $s=\lceil\sqrt n\rceil$-element subsets of $[n]$, independently of the stream. The uniform size-$s$ sample can be generated sequentially using $O(\log n)$ auxiliary bits \cite{10.1145/3147.3165}, in addition to the $O(s\log n)$ bits used to store its labels.

Let
\[
 d=\min\left\{n,\ \left\lceil\frac ns\right\rceil
           \left\lceil\log\frac n\varepsilon\right\rceil\right\}.
\]
Cap the buffer at $nd$ edges. On overflow, the decision algorithm returns a YES-instance and halts. The coloring variant instead returns failure.

\begin{lemma}[Buffer size]\label{lem:buffer}
For every fixed edge order, $|B|\le nd$ with probability at least $1-\varepsilon$.
\end{lemma}
\begin{proof}

The proof is trivial when $d=n$, thus we assume $d<n$.
For each vertex $v$, list its distinct neighbors in stream order. Let $R_v$ be the number of incident edges arriving while $v$ is still unassigned \emph{before} that edge. This includes its first assignment edge if one exists. If $R_v>d$, then $v$ has more than $d$ neighbors and none of its first $d$ neighbors is an anchor. Uniform sampling without replacement gives
\[
 \Prb[R_v>d]
 \le\frac{\binom{n-d}{s}}{\binom ns}
 \le(1-d/n)^s
 \le e^{-sd/n}
 \le\frac\varepsilon n.
\]

By a union bound, with probability at least $1-\varepsilon$, we have $R_v\le d$ for every vertex $v$. In this case, every buffered edge is counted by at least one $R_v$, and therefore $|B|\le \sum_v R_v\le nd$.
\end{proof}

We now bound the size of the forests.
Each cross-group forest contains at most $|A_a|+|A_b|$ edges and each internal forest at most $|A_a|$. Consequently
\[
\begin{aligned}
\sum_{a\le b}|F_{ab}|
&\le
\sum_{a<b}(|A_a|+|A_b|)
+\sum_a|A_a|\\
&=(s-1)\sum_a|A_a|+\sum_a|A_a|\\
&=s\sum_a|A_a|\\
&\le sn.
\end{aligned}
\]
Since running time is unrestricted, we can check connectivity and path parity by searching the stored forests, reusing the same $O(n)$ temporary entries for every search. Forest headers require $O(s^2)=O(n)$ additional entries. Each entry uses $O(\log n)$ bits.

Thus, the algorithm uses:
\[
 O\bigl((ns+nd+n)\log n\bigr)
 =O\bigl(n^{3/2}\log(n/\varepsilon)\log n\bigr) \text{ bits of space.}
\]

Theorem \ref{thm:core} now easily follows.

\section{ $\chi(G)\le3$ vs $\chi(G)>Q$ }\label{sec:generalupperbound}

We use the following immediate consequence\footnote{The only difference between our stated lemma and the original lemma is that, in the original, the probability must be exactly $2\ln(n)/L$, whereas in our version it only needs to be at least $2\ln(n)/L$.} of~\cite[Lemma~8.2 (arXiv version)]{assadi2026coloring}.

\begin{lemma}\label{lem:vertex}
 Let $G$ be an $n$-vertex graph and $H=G[S]$ be an induced subgraph of $G$ obtained by sampling each vertex of $G$ independently with probability  $p\in(0,1)$ such that $L$ is a positive integer and $pL\ge 2\ln n$ then
\[
    \Prb\left[
        \chi(G[S])<\frac{\chi(G)}{L}-1
    \right]
    \le \frac12.
\]
\end{lemma}

The proof of Lemma \ref{lem:vertex} can be found in Appendix \ref{app:vertex-sampling}.
 
\begin{theorem}\label{thm:gap}
Let $n\ge2$, let $c\ge1$ be a fixed integer, and let
$\ell=\lceil\log_2 n\rceil$. For every integer $Q$ with
$144\ell\le Q\le n$, there is a one pass adversarial-order
insertion-only algorithm distinguishing $\chi(G)\le3$ from
$\chi(G)>Q$ with success probability at least $1-n^{-c}$, using
\[
  O\!\left((n/Q)^{3/2}\log^5 n\right)
\]
bits in the worst case.
\end{theorem}

\begin{proof}
\noindent
Set $\delta=n^{-c}$ and use the following parameters:
\begin{align*}
  b&=\left\lfloor\log_2\frac{Q}{72\ell}\right\rfloor,
  &p&=2^{-b},
  &L&=4\ell\,2^b,\\
  R&=\left\lceil\log_2\frac4\delta\right\rceil,
  &\mu&=np,
  &K&=\lceil2\mu\rceil+
       \left\lceil\log_2\frac{4R}{\delta}\right\rceil.
\end{align*}
The assumption $Q\ge144\ell$ gives $b\ge1$ and
\[
  \frac{72\ell}{Q}\le p<\frac{144\ell}{Q},
  \qquad
  pL=4\ell\ge2\ln n,
  \qquad
  Q/L\ge18.
\]

\medskip
Prepare $R$ independent Bernoulli vertex samples.
For each vertex $v$ and each trial $i$, generate $b$
independent random bits, each equally likely to be $0$ or $1$.
Include $v$ in $S_i$ if and only if all $b$ bits are $0$,
which occurs with probability $2^{-b}=p$.
The bits used for different vertices and trials are also
independent.

Let each trial store at most $K$ selected vertices. As soon as a trial selects
a $(K+1)$-st vertex, discard its stored sample and let that
trial return a YES-instance. Let $S_i$ denote the full
Bernoulli sample defined by the random selections, even
if the actual trial is discarded before its sample is
completed.

A trial with $|S_i|\le1$ returns a YES-instance.
For every remaining trial, with $2\le|S_i|\le K$, run
the algorithm of Theorem~\ref{thm:core} on $G[S_i]$ with
error parameter $\varepsilon=\frac{\delta}{4R}$.

Each such instance chooses its own anchors within $S_i$, independent of the other instances.
When an original edge arrives, send it to each active
instance whose sample contains both endpoints.
Return a NO-instance in the original input if any trial returns a NO-instance, and
otherwise return a YES-instance.

\noindent\textbf{Correctness.}
If $\chi(G)\le3$, every sampled induced subgraph is
3-colorable.
Discarded trials and trials with at most one sampled
vertex return a YES-instance, while every active trial
returns a YES-instance by the perfect completeness of
Theorem~\ref{thm:core}.
Thus, the overall algorithm has perfect completeness.

Suppose now that $\chi(G)>Q$.
Call $S_i$ a witness if
\[
  \chi(G[S_i])\ge\frac{\chi(G)}{L}-1.
\]
Every witness satisfies
\[
  \chi(G[S_i])
  >\frac{Q}{L}-1
  \ge17
  >6.
\]
Let $E_1$ be the event that no sample is a witness,
$E_2$ the event that some sample has more than $K$
vertices\footnote{We need to account for $E_2$ as it could lead to a witness being discarded.}, and $E_3$ the event that some surviving
witness trial returns a YES-instance.

By Lemma~\ref{lem:vertex} and independence of the
samples,
\[
  \Prb[E_1]\le2^{-R}\le\frac{\delta}{4}.
\]
For each trial, $Z=|S_i|$ is binomial with mean
$\mu$, so $\mathbb{E}[2^Z]=(1+p)^n\le e^\mu$.
Consequently and by Markov's inequality,
\[
  \Prb[Z>K]
  \le e^\mu 2^{-K}
  \le e^{(1-2\ln2)\mu}\frac{\delta}{4R}
  \le\frac{\delta}{4R}.
\]
A union bound gives $\Prb[E_2]\le\delta/4$.

Since $\chi(G[T])>6$, Theorem~\ref{thm:core} bounds
its false-acceptance probability by
$\varepsilon=\delta/(4R)$.
Averaging over $T$ and taking a union bound over
trials gives $\Prb[E_3]\le\delta/4$.

Outside $E_1\cup E_2\cup E_3$, a witness survives
and its trial returns a NO-instance.
Hence
\[
  \Prb[\text{false acceptance}]
  \le\Prb[E_1]+\Prb[E_2]+\Prb[E_3]
  \le\frac{3\delta}{4}
  <\delta.
\]

\medskip
\noindent\textbf{Space complexity.} Since $p<144\ell/Q$ and $c$ is fixed,
\[
  \mu=np=O((n/Q)\log n),
  \qquad
  R=O(\log n),
  \qquad
  K=O(n/Q\log n).
\]
Each trial stores at most $K$ sampled vertices.
An active instance of the algorithm of
Theorem~\ref{thm:core}, using original identifiers of
$O(\log n)$ bits and error parameter
$\varepsilon=\delta/(4R)$, uses
\[
  O((n/Q)^{3/2}\log^{7/2} n)
\]
bits.
Multiplying by the $R=O(\log n)$ concurrent trials gives
\[
  O\!\left(RK^{3/2}\log^2 n\right)
  =
  O\!\left((n/Q)^{3/2}\log^{9/2}n\right)
  \subseteq
  O\!\left((n/Q)^{3/2}\log^5 n\right).
\]

The membership lists use
\[
  O(RK\log n)=O((n/Q)\log^3 n)
\]
bits, which fit within this bound.

\end{proof}

\begin{corollary}\label{cor:main}
For every fixed integer $c\ge1$ and all sufficiently
large $n$, there is a one pass adversarial-order
semi-streaming decision algorithm distinguishing
\[
  \chi(G)\le3
  \qquad\text{from}\qquad
  \chi(G)>\lceil n^{1/3}\rceil,
\]
with success probability at least $1-n^{-c}$,
using $O(n\log^5 n)$ bits.
\end{corollary}

\section*{Acknowledgments}

I would like to thank Santhoshini Velusamy for introducing me to streaming algorithms and for her help and support throughout the process. 

Deep Review (by Google) was used on a draft of this manuscript. Multiple versions of ChatGPT were used for literature search and writing assistance. Theorem \ref{thm:gap} and Appendix \ref{app:vertex-sampling} were entirely obtained using GPT Astra; the author subsequently improved the presentation and takes no credit, but full responsibility, for these results.

\appendix

\section{Proof of the Vertex-Sampling Lemma}
\label{app:vertex-sampling}

For the sake of completeness, we give the proof of
Lemma~\ref{lem:vertex}.

\begin{proof}
Set
\[
    p_0=\frac{2\ln n}{L}.
\]
Since $pL\ge 2\ln n$, we have $p_0\le p$. Moreover,
\[
    \frac{2\ln n}{p_0}=L\in\mathbb{N}.
\]

Let $T$ be obtained by sampling each vertex independently with
probability $p_0$. By~\cite[Lemma~8.2]{assadi2026coloring},
\[
    \Prb\left[
        \chi(G[T])<\frac{\chi(G)}{L}-1
    \right]
    \le \frac12.
\]

Couple the samples so that $T\subseteq S$: first sample $S$ with
probability $p$, and then retain each vertex of $S$ independently
with probability $p_0/p$ to obtain $T$. Since $T\subseteq S$,
\[
    \chi(G[T])\le \chi(G[S]).
\]
Therefore,
\[
    \left\{
        \chi(G[S])<\frac{\chi(G)}{L}-1
    \right\}
    \subseteq
    \left\{
        \chi(G[T])<\frac{\chi(G)}{L}-1
    \right\},
\]
and hence
\[
    \Prb\left[
        \chi(G[S])<\frac{\chi(G)}{L}-1
    \right]
    \le \frac12.
\]
\end{proof}

\bibliographystyle{siam}
\bibliography{references}

\end{document}